\documentclass{svmult}

\renewcommand{\email}[1]{\emailname: #1} 

\usepackage{mathptmx}       
\usepackage{helvet}         
\usepackage{courier}        

\usepackage{makeidx}         
\usepackage{graphicx}        
\usepackage[bottom]{footmisc}

\usepackage{latexsym}
\usepackage{amsmath}
\usepackage{amsfonts}
\usepackage{amssymb}
\usepackage{bm}

\usepackage{url}
\usepackage{algorithm}
\usepackage{algorithmic}
\usepackage[misc,geometry]{ifsym}

\renewenvironment{proof}{\noindent{\itshape Proof.}}{\smartqed\qed}
 \newenvironment{proofof}[1]{\noindent{\itshape Proof of {#1}.}}{\smartqed\qed}

\spdefaulttheorem{assumption}{Assumption}{\upshape \bfseries}{\itshape}
\spdefaulttheorem{algo}{Algorithm}{\upshape \bfseries}{\itshape}

\newcommand{\bsa}{{\boldsymbol{a}}}
\newcommand{\bsb}{{\boldsymbol{b}}}

\newcommand{\bsx}{{\boldsymbol{x}}}
\newcommand{\bsy}{{\boldsymbol{y}}}

\DeclareSymbolFont{bbold}{U}{bbold}{m}{n}
\DeclareSymbolFontAlphabet{\mathbbold}{bbold}

\begin{document}

\title*{An Upper Bound of the Minimal Dispersion via Delta Covers}

\author{Daniel Rudolf}

\institute{Daniel Rudolf\\
Institut f\"ur Mathematische Stochastik, University of Goettingen, 
Goldschmidtstra\ss e 7,
37077 G\"ottingen, Germany
\\
 \email{daniel.rudolf@uni-goettingen.de}
}

\maketitle

\index{Rudolf, Daniel}

\paragraph{Dedicated to Ian H.~Sloan on the occasion of his 80th birthday.}

\abstract{For a point set of $n$ elements in the $d$-dimensional unit cube and 
a class of test sets 
we are interested in 
the largest volume of a test set
which does not contain any point. 
For all natural numbers $n$, $d$ 
and under the assumption of the existence of a $\delta$-cover with cardinality 
$\vert \Gamma_\delta \vert$ 
we prove that there is a point set,
such that the largest volume of such a test set without any point is bounded above by 
$\frac{\log \vert \Gamma_\delta \vert}{n} + \delta$.
For axis-parallel boxes on the unit cube this leads to
a volume of at most $\frac{4d}{n}\log(\frac{9n}{d})$ and on the torus to
$\frac{4d}{n}\log (2n)$. 
}

\section{Introduction and Main Results}
For a point set $P$ of $n$ elements in the unit cube $[0,1]^d$ 
and for a set $\mathcal{B}$ 
of measurable subsets of $[0,1]^d$ the quantity of interest is the
\emph{dispersion},
given by
\begin{equation}
\label{eq: disp}
 {\rm disp}(P,\mathcal{B}) := \sup_{P\cap B= \emptyset,\, B\in\mathcal{B} } \lambda_d(B).
\end{equation}
Here $\lambda_d$ denotes the $d$-dimensional Lebesgue measure and $\mathcal{B}$ is called set of test sets. 
The dispersion measures the size 
of the largest hole which does not contain any point of $P$. 
The shape of the hole is specified by the set of test sets. 
We are interested in point sets with best possible upper bounds 
of the dispersion, which thus allow only small holes without any point. Of course, any estimate of 
${\rm disp}(P,\mathcal{B})$
depends
on $n$, $d$ and $\mathcal{B}$.

Classically, the dispersion of a point set $P$ 
was introduced by Hlawka \cite{Hl76} as the radius
of the largest ball, with respect to some metric, which does not contain any point of $P$. 
This quantity appears in the setting of 
quasi-Monte Carlo methods for optimization, see \cite{Ni83} and \cite[Chapter 6]{Ni92}.
The notion of the dispersion from
\eqref{eq: disp} was 
introduced by Rote and Tichy in \cite{RoTi96} to allow more general test sets. 
There the focus is on the dependence of $n$ (the cardinality of the point set) 
of ${\rm disp}(P,\mathcal{B})$. 
In contrast to that, we are also interested 
in the behavior with respect to the dimension.

There is 
a well known relation 
to the star-discrepancy, namely, the dispersion is a lower bound of this quantity.
For further literature, open problems, recent developments and applications 
related to this topic
we refer to \cite{DiPi10,DiRuZh13,Ni92,No15,NoWo10}.

For the test sets we focus on axis-parallel boxes.
Point sets with small dispersion with respect to such
axis-parallel boxes are 
useful for the approximation of rank-one tensors, see \cite{BaDaDeGr14,NoRu16}.
In computational geometry, given a point configuration the problem 
of finding the largest empty axis-parallel box is well studied. 
Starting with \cite{NaLeHs84}
for $d=2$, there is a considerable 
amount of work for $d>2$, see \cite{DuJi13,DuJi16} and
the references therein. 
Given a large dataset of points, the search for empty axis-parallel boxes
is motivated by the fact that such boxes 
may reveal natural constraints in the data and thus unknown correlations, see \cite{EdGrLiMi03}.

The \emph{minimal dispersion}, given by 
\[
 {\rm disp}_{\mathcal{B}}(n,d) := \inf_{P\subset [0,1]^d, \vert P \vert=n} {\rm disp}(P,\mathcal{B}),
\]
quantifies 
the best possible behavior of the dispersion with respect to $n$, $d$ and $\mathcal{B}$.
Another significant quantity is the \emph{inverse of the minimal dispersion}, that is, the minimal
number of points $N_\mathcal{B}(d,\varepsilon)$ with minimal dispersion at most $\varepsilon\in(0,1)$, i.e.,
\[
 N_\mathcal{B}(d,\varepsilon) = \min\{n\in\mathbb{N}\mid {\rm disp}_{\mathcal{B}}(n,d) \leq \varepsilon\}.
\]
By virtue of a result of Blumer, Ehrenfeucht, Haussler and Warmuth \cite[Lemma~A2.1, Lemma~A2.2 and Lemma~A2.4]{BlEhHaWa89} 
one obtains 
\begin{equation}  \label{eq: vc_est}
  {\rm disp}_{\mathcal{B}}(n,d) \leq \frac{2 d_\mathcal{B}}{n} \log_2\Big(\frac{6 n}{d_\mathcal{B}}\Big) 
  \quad \text{for} \quad
  n\geq d_\mathcal{B},
\end{equation}
or stated differently
\begin{equation}  \label{eq: N_vc_inv}
  N_\mathcal{B}(d,\varepsilon) \leq 8 d_\mathcal{B} \varepsilon^{-1} \log_2(13 \varepsilon^{-1}),
\end{equation}
where $\log_2$ is the dyadic logarithm and $d_\mathcal{B}$ denotes the VC-dimension\footnote{The VC-dimension is the cardinality
of the largest subset $T$ of $[0,1]^d$ such that the set system $\{ T\cap B\mid B\in \mathcal{B} \}$
contains all subsets of $T$. } of $\mathcal{B}$. 
The dependence on $d$ is hidden in the VC-dimension $d_\mathcal{B}$. For
example, for the set of test sets of axis-parallel boxes 
\[
 \mathcal{B}_{\rm ex} = \{ \Pi_{k=1}^d [x_k,y_k) \subseteq [0,1]^d \mid x_k < y_k,\; k=1,\dots,d \},
\]
it is well known that $d_{\mathcal{B}_{\rm ex}} = 2d$. However, the concept of VC-dimension is not as easy
to grasp as it might seem on the first glance and it is also not trivial to prove upper bounds on $d_\mathcal{B}$
depending on $\mathcal{B}$. For instance, for \emph{periodic axis-parallel boxes}, 
which coincide with the interpretation of
considering the torus instead of the unit cube, given by
\[
  \mathcal{B}_{\rm per} = \{ \Pi_{k=1}^d I_k(\bsx,\bsy) \mid  \bsx=(x_1,\dots,x_d),\bsy=(y_1,\dots,y_d)\in [0,1]^d \}
\]
with
\[
 I_k(\bsx,\bsy) = \begin{cases}
             (x_k,y_k) & x_k<y_k \\
             [0,1]\setminus [y_k,x_k] & y_k \leq x_k,
            \end{cases}
\]
the dependence on $d$ in $d_{\mathcal{B}_{\rm per}}$ is not obvious. The conjecture here is that $d_{\mathcal{B}_{\rm per}}$ behaves similar as 
$d_{\mathcal{B}_{\rm ex}}$, i.e., linear in $d$, but we do not have a proof for this fact.

The aim of this paper is to prove an estimate similar to 
\eqref{eq: vc_est} based on the concept of a $\delta$-cover of $\mathcal{B}$.
For a discussion about
$\delta$-covers, bracketing numbers and VC-dimension we refer to \cite{Gn08}.
%
%
%
%
%
%
Let $\mathcal{B}$ be a set of measurable subsets of $[0,1]^d$.
A \emph{$\delta$-cover
for $\mathcal{B}$} with $\delta>0$ is a finite set $\Gamma_\delta \subseteq \mathcal{B}$ 
which satisfies
\[
 \forall B\in \mathcal{B}\quad \exists L_B,U_B\in \Gamma_\delta \quad \text{with} \quad L_B \subseteq B \subseteq U_B
\]
such that $\lambda_d(U_B\setminus L_B) \leq \delta$. 
The main abstract theorem is as follows.

\begin{theorem} \label{thm: mainthm}
 For a set of test sets $\mathcal{B}$ assume that for $\delta>0$ the set $\Gamma_\delta$ is a $\delta$-cover
 of $\mathcal{B}$. Then
\begin{equation} \label{eq: delta_cov_est}
 {\rm disp}_\mathcal{B}(n,d) \leq 
 \frac{\log\vert{\Gamma_\delta}\vert}{n}+\delta. 
\end{equation}
\end{theorem}

The cardinality of the $\delta$-cover plays a crucial role in the upper bound of
the minimal dispersion.
Thus, to apply the theorem 
to concrete sets of test sets 
one has to construct suitable, not too large, $\delta$-covers. 

For $\mathcal{B}_{\rm ex}$
the best results on $\delta$-covers we know are due to Gnewuch, see \cite{Gn08}.
As a consequence of the theorem and a combination of
\cite[Formula (1), Theorem~1.15, Lemma~1.18]{Gn08} one obtains
\begin{corollary}  \label{cor: B_ex}
 For $\mathcal{B}_{\rm ex}$ and $n>2d$ we have
 \begin{equation}  \label{eq: B_ex}
     {\rm disp}_{\mathcal{B}_{\rm ex}}(n,d) \leq \frac{4d}{n} 
    \log\Big(\frac{9 n}{d}\Big). 
 \end{equation}
%
%
%
(For $n\leq 2d$ the trivial estimate ${\rm disp}_{\mathcal{B}_{\rm ex}}(n,d)\leq 1$ applies.)
In particular,
\begin{equation} \label{eq: N_B_ex}
 N_{\mathcal{B}_{\rm ex}}(\varepsilon,d)
 \leq 8d \varepsilon^{-1} \log(33\varepsilon^{-1}). 
\end{equation}
\end{corollary}
Obviously, this is essentially the same
as the estimates \eqref{eq: vc_est} and \eqref{eq: N_vc_inv}
in the setting of $\mathcal{B}_{\rm ex}$. 
Let us discuss how those estimates fit into the literature. 
From \cite[Theorem~1 and (4)]{AiHiRu15} we know that
\begin{equation} \label{eq: low_ub_Bex}
 \frac{\log_2 d }{4(n+\log_2d)}\leq  {\rm disp}_{\mathcal{B}_{\rm ex}}(n,d) \leq \frac{1}{n}
 \min\Big\{
 2^{7d+1},2^{d-1} \Pi_{i=1}^{d-1} p_i  \Big\},
\end{equation}
where $p_i$ denotes the $i$th prime.
The upper bound $2^{7d+1}/n$ is due to Larcher based on suitable $(t,m,d)$-nets 
and for $d\geq 54$ improves the super-exponential 
estimate $2^{d-1} \Pi_{i=1}^{d-1} p_i/n$ 
of Rote and Tichy \cite[Proposition~3.1]{RoTi96} based on the Halton sequence.
The order of convergence with respect to $n$ is optimal, but the dependence 
on $d$ in the upper bound is exponential.  In
the estimate of Corollary~\ref{cor: B_ex} the optimal order in $n$ is not achieved, but 
the dependence on $d$ is much better.
Already for $d=5$ it is required that
 $n$ must be larger than $5 \cdot 10^{72}$ to obtain a smaller upper bound
from \eqref{eq: low_ub_Bex} than from \eqref{eq: B_ex}.
By rewriting the result of Larcher in terms of 
$N_{\mathcal{B}_{\rm ex}}(\varepsilon,d)$ the dependence on $d$ can be very well illustrated, one obtains
\[
 N_{\mathcal{B}_{\rm ex}}(\varepsilon,d) \leq 2^{7d+1} \varepsilon^{-1}.
\]
Here, for fixed $\varepsilon$ there is an exponential dependence on $d$, whereas in the estimate of 
\eqref{eq: N_B_ex} there is a linear dependence on $d$. Summarizing, according to 
$ N_{\mathcal{B}_{\rm ex}}(\varepsilon,d)$
the result of Corollary~\ref{cor: B_ex} reduces 
the gap with respect to $d$, we obtain\footnote{After acceptance of the current paper a new upper bound
 of $N_{\mathcal{B}_{\rm ex}}(\varepsilon,d)$ was proven in \cite{So17}.
 From \cite{So17} one obtains
 for $\varepsilon\in(0,1/4)$ that
 \[
 N_{\mathcal{B}_{\rm ex}}(\varepsilon,d) \leq c_\varepsilon \log_2 d 
 \]
 with $c_\varepsilon = \varepsilon^{-(\varepsilon^{-2}+2)}(4 \log \varepsilon^{-1}+1)$
 for $\varepsilon^{-1}\in \mathbb{N}$.
 In particular, it shows that the lower bound cannot be improved with respect to the dimension.
 Note that the dependence on $\varepsilon^{-1}$ is not as good as in \eqref{eq: N_B_ex}.}
\[
 (1/4-\varepsilon) \varepsilon^{-1} \log_2 d \leq
  N_{\mathcal{B}_{\rm ex}}(\varepsilon,d)
 \leq 8d \varepsilon^{-1} \log(33\varepsilon^{-1}).
\]
As already mentioned for $\mathcal{B}_{\rm per}$ the estimates 
\eqref{eq: vc_est} and \eqref{eq: N_vc_inv} are not applicable, 
since we do not know the VC-dimension.
We construct a $\delta$-cover 
in Lemma~\ref{lem: delta_cov_B_per} below and obtain the following estimate
as a consequence
of the theorem. 
Note that, since $\mathcal{B}_{\rm ex} \subset \mathcal{B}_{\mathcal{\rm per}}$, we cannot expect 
something better than in Corollary~\ref{cor: B_ex}.

\begin{corollary} \label{cor: cor_B_per}
 For $\mathcal{B}_{\rm per}$ and $n\geq 2$ we have
 \begin{equation} \label{eq: disp_B_per}
    {\rm disp}_{\mathcal{B}_{\rm per}}(n,d) \leq 
      \frac{4d}{n} 
     \log(2 n).  
 \end{equation}
 In particular,
 \begin{equation} \label{eq: N_inv_B_per}
  N_{\mathcal{B}_{\rm per}}(\varepsilon,d) \leq 8d \varepsilon^{-1}
  [\log(8d) + \log\varepsilon^{-1}].
 \end{equation}
\end{corollary} 
Indeed, the estimates of Corollary~\ref{cor: cor_B_per} are not as good 
as the estimates of Corollary~\ref{cor: B_ex}. 
By adding the result of Ullrich \cite[Theorem~1]{Ul15}
one obtains
\begin{equation*} 
 \min\{1,d/n\} \leq {\rm disp}_{\mathcal{B}_{\rm per}}(n,d)\leq \frac{4d}{n} 
     \log(2 n) ,
\end{equation*}
or stated differently,
\begin{equation} \label{eq: low_up_B_per}
  d \varepsilon^{-1} \leq N_{\mathcal{B}_{\rm per}}(\varepsilon,d) \leq 
  8d \varepsilon^{-1}
  [\log(8d) + \log\varepsilon^{-1}]. 
\end{equation}
In particular, \eqref{eq: low_up_B_per} illustrates the dependence on the dimension, namely,
for fixed $\varepsilon\in (0,1)$ Corollary~\ref{cor: cor_B_per} gives, except 
of a $\log d$ term, the right dependence on $d$.

In the rest of the paper we prove the stated results and provide a conclusion.

\section{Auxiliary Results, Proofs and Remarks}
For the proof of Theorem~\ref{thm: mainthm} we need the following lemma.
\begin{lemma} \label{lem: aux_lem}
For $\delta>0$ let $\Gamma_\delta$ be a $\delta$-cover of $\mathcal{B}$. Then, for any point set $P\subset [0,1]^d$
with $n$ elements we have
\[
 {\rm disp}(P,\mathcal{B}) \leq \delta+ \max_{ A\cap P = \emptyset,\;A\in \Gamma_\delta} \lambda_d(A).
\]
\end{lemma}
\begin{proof}
 Let $B\in \mathcal{B}$ with $B\cap P = \emptyset$. Then, 
 there are $L_B,U_B\in \Gamma_\delta $ with $L_B\subseteq B \subseteq U_B$
 such that
 \[
  \lambda_d(B\setminus L_B) \leq \lambda_d(U_B\setminus L_B) \leq \delta.
 \]
 In particular, $L_B\cap P = \emptyset$ and 
 \[
   {\rm disp}(P,\mathcal{B}) \leq 
   \sup_{P\cap B= \emptyset,\, B\in\mathcal{B} } \left( \lambda_d(U_B\setminus L_B) + \lambda_d(L_B) \right)
   \leq \delta + \max_{ A\cap P = \emptyset,\;A\in \Gamma_\delta} \lambda_d(A).  
 \]
\end{proof}
\begin{remark}
 In the proof we actually only used that there is a 
 set $L_B\subseteq B$ with $\lambda_d(B\setminus L_B) \leq \delta$. 
 Thus, instead of considering $\delta$-covers
 it would be enough to work with set systems which approximate $B$ from below 
 up to $\delta$.
\end{remark}

By probabilistic arguments similar to those of \cite[Section~8.1]{BeCh87}
we prove the main theorem. 
As in \cite[Theorem~1 and Theorem~3]{HeNoWaWo01} for the star-discrepancy, it also turns out that 
such arguments are useful for studying the dependence on the dimension of the dispersion.

\begin{proofof}{Theorem~\ref{thm: mainthm}}
 By Lemma~\ref{lem: aux_lem} it is enough to show that there is a point set $P$ 
 which satisfies 
 \begin{equation} \label{eq: good_prop}
    \max_{ A\cap P = \emptyset,\;A\in \Gamma_\delta} \lambda_d(A) \leq 
    \frac{\log \vert{\Gamma_\delta}\vert}{n}.
 \end{equation}
 Let $(\Omega, \mathcal{F},\mathbb{P})$ be a probability space and 
 $(X_i)_{1\leq i \leq n}$ be an iid sequence of uniformly
 distributed 
 random variables mapping from $(\Omega, \mathcal{F},\mathbb{P})$ into $[0,1]^d$. 
 We consider the sequence of random variables as ``point set''
 and prove that with high probability the desired property \eqref{eq: good_prop}
 is satisfied.
 For $(c_n)_{n\in \mathbb{N}}\subset (0,1)$ we have
 \begin{align*}
   \mathbb{P}\Big( &\max_{A\in \Gamma_\delta,\; A\cap\{X_1,\dots,X_n\}=\emptyset} 
  \lambda_d(A)\leq c_n\Big)
  = \mathbb{P} 
  \Big( 
\bigcap_{A\in \Gamma_\delta} 
\{ 
  \mathbf{1}_{A\cap\{X_1,\dots,X_n\}=\emptyset} \cdot
  \lambda_d(A) \leq c_n \}  \Big)\\
  &  = 1-\mathbb{P}\Big(\bigcup_{A\in \Gamma_\delta} 
    \{ 
  \mathbf{1}_{A\cap\{X_1,\dots,X_n\}=\emptyset} \cdot
  \lambda_d(A) > c_n \} \Big)\\
 &  \geq 1 -\sum_{A\in \Gamma_\delta} 
  \mathbb{P}\Big( \mathbf{1}_{A\cap\{X_1,\dots,X_n\}=\emptyset} \cdot \lambda_d(A)>c_n\Big)\\
  & > 1 - \vert{\Gamma_\delta}\vert(1-c_n)^n.
 \end{align*}
 By the fact that $1-\vert{\Gamma_\delta}\vert^{-1/n} \leq \frac{\log\vert{\Gamma_\delta}\vert}{n}$
 and by choosing $c_n =  \frac{\log\vert{\Gamma_\delta}\vert}{n}$ we obtain
 \[
  \mathbb{P}\Big( \max_{A\in \Gamma_\delta,\; A\cap\{X_1,\dots,X_n\}=\emptyset} 
  \lambda_d(A)\leq \frac{\log\vert{\Gamma_\delta}\vert}{n}\Big)>0.
 \]
 Thus, there exists a realization of $(X_i)_{1\leq i \leq n}$, say 
 $({\bsx}_i)_{1\leq i\leq n} \subset [0,1]^d$, so that for $P=\{\bsx_1,\dots,\bsx_n\}$ the 
 inequality \eqref{eq: good_prop} is satisfied.
\end{proofof}

\begin{remark} \label{rem}
 By Lemma~\ref{lem: aux_lem} and the same arguments as in the proof of the theorem
 one can see that a point set of iid uniformly distributed random variables $X_1,\dots,X_n$
 satisfies a ``good dispersion bound'' with high probability. In detail,
 \begin{align*}
  \mathbb{P}\left({\rm disp}(\{X_1,\dots,X_n\},\mathcal{B})\leq 2\delta\right)
  & \geq    \mathbb{P}\Big( \max_{A\in \Gamma_\delta,\; A\cap\{X_1,\dots,X_n\}=\emptyset} 
  \lambda_d(A)\leq \delta\Big) \\
  & > 1-\vert{\Gamma_\delta}\vert (1-\delta)^n.
 \end{align*}
 In particular, for confidence level $\alpha\in (0,1]$ and
 \[
  n := \frac{\log(\vert{\Gamma_\delta}\vert\alpha^{-1})}{\delta} 
  \geq 
  \frac{\log(\vert{\Gamma_\delta}\vert\alpha^{-1})}{\log(1-\delta)^{-1}} 
 \]
 the probability that the random point set has dispersion smaller than $2\delta$
 is strictly larger than $1-\alpha$.
This implies
 \begin{equation}  \label{eq: N_delta_cover}
  N_\mathcal{B}(d,\varepsilon) \leq 2 \varepsilon^{-1} \log\vert{\Gamma_{\varepsilon/2}}\vert,
 \end{equation}
 where the dependence on $d$ is hidden in $\vert{\Gamma_{\varepsilon/2}}\vert$.
 \end{remark}


In the spirit of \cite{NoWo08,NoWo10,NoWo12} we are interested in 
\emph{polynomial tractability}
of the minimal dispersion, that is, $N_\mathcal{B}(d,\varepsilon)$
may not grow faster than polynomial in $\varepsilon^{-1}$ and $d$.
 The following corollary is a consequence of the theorem
and provides a condition on the $\delta$-cover for such polynomial tractability.

\begin{corollary}  \label{cor1}
For $\delta\in (0,1)$ and the set of test sets $\mathcal{B}$ let 
 $\Gamma_\delta$ be a $\delta$-cover satisfying 
 \[
  \exists c_1\geq 1\; \& \;c_2,c_3 \geq 0 
  \quad \text{s.t.}
  \quad \vert{\Gamma_\delta}\vert \leq (c_1 d^{c_2} \delta^{-1} )^{c_3 d}.
 \]
Then, for $n>c_3d$ one has
\[
  {\rm disp}_\mathcal{B}(n,d) \leq \frac{c_3d}{n}  \Big[ 
   \log\Big( \frac{ c_1 d^{c_2-1} n}{c_3}  \Big)+1 \Big].
\]
\end{corollary}
\begin{proof}
 Set $\delta = c_3 d/n$ in \eqref{eq: delta_cov_est} and the assertion follows. 
\end{proof}

This implies the result of Corollary~\ref{cor: B_ex}.

\begin{proofof}{Corollary~\ref{cor: B_ex}}
 By \cite[Formula (1), Theorem~1.15, Lemma~1.18]{Gn08} one has
 \[
  \vert{\Gamma_\delta}\vert \leq \frac{1}{2} (2\delta^{-1}+1)^{2d} \cdot \frac{(2d)^{2d}}{(d!)^2} 
  \leq (6\mathrm{e}\delta^{-1})^{2d}.
 \]
 Here the last inequality follows mainly by $d!>\sqrt{2\pi d}(d/\mathrm{e})^d$
 and the assertion is proven by Corollary~\ref{cor1} with $c_1=6\mathrm{e}$, $c_2=0$, $c_3=2$.
\end{proofof}

For $\mathcal{B}_{\rm per}$ we need to construct a $\delta$-cover.
\begin{lemma}  \label{lem: delta_cov_B_per}
For $\mathcal{B}_{\rm per}$ with $\delta>0$ and $m=\lceil 2d/\delta \rceil$ the set
\[
 \Gamma_\delta = \left\{ \Pi_{k=1}^d I_k(\bsa,\bsb)\mid \bsa,\bsb\in G_m \right\}
\]
with
\[
 G_m = \{ (a_1,\dots,a_d)\in [0,1]^d \mid a_k = i/m, \; i=0,\dots,m;\; k=1,\dots,d \}
\]
is a $\delta$-cover
and satisfies $\vert{\Gamma_\delta}\vert = (m+1)^{2d}$.
\end{lemma}
\begin{proof}
For arbitrary $\bsx,\bsy\in [0,1]^d$ with $\bsx=(x_1,\dots,x_d)$ and
$\bsy=(y_1,\dots,y_d)$ there are 
\begin{align*}
 \bsa& =(a_1,\dots,a_d)\in G_m,\quad
 \bar{\bsa} =(\bar{a}_1,\dots,\bar{a}_d)\in G_m\\
  \bsb& =(b_1,\dots,b_d)\in G_m,\quad
 \bar{\bsb} =(\bar{b}_1,\dots,\bar{b}_d)\in G_m,
 \end{align*}
such that
\begin{align*}
 a_k & \leq x_k \leq \bar{a}_k\leq a_k+1/m,\qquad
 b_k \leq y_k \leq \bar{b}_k \leq b_k+1/m.
\end{align*}
Define $B(\bsx,\bsy)=\Pi_{k=1}^d I_k(\bsx,\bsy)$ 
and note that it is enough to find $L_B, U_B\in \Gamma_\delta$ 
with $L_B \subseteq B(\bsx,\bsy)\subseteq U_B$ and $\lambda_d(U_B\setminus L_B)\leq \delta$.
For any coordinate $k\in \{1,\dots,d\}$ we distinguish four cases illustrated in Figure~\ref{fig:cases}:
\begin{enumerate}
 \item Case: $\vert{x_k-y_k}\vert\leq 1/m$ and $x_k<y_k$:\\
 Define $I_k^L=\emptyset$ and $I_k^U=(a_k,\bar{b}_k)$. 
 (Here $I_k^L=[0,1]\setminus[0,1]=\emptyset$.)
 \item Case: $\vert{x_k-y_k}\vert\leq 1/m$ and $x_k\geq y_k$:\\
 Define $I_k^L=[0,1]\setminus [b_k,\bar{a}_k]$ and 
 $I_k^U=[0,1]\setminus[a_k,a_k]$. (Here $I_k^U=[0,1]\setminus\{a_k\}$.)
 \item Case: $\vert{x_k-y_k}\vert> 1/m$ and $x_k<y_k$:\\
 Define $I_k^L=(\bar{a}_k,b_k)$ and $I_k^U=(a_k,\bar{b}_k)$.
  \item Case: $\vert{x_k-y_k}\vert> 1/m$ and $x_k\geq y_k$:\\
  Define $I_k^L=[0,1]\setminus[b_k,\bar{a}_k]$ and $I_k^U=[0,1]\setminus[\bar{b}_k,a_k]$. 
\end{enumerate}
\begin{figure}    
    \qquad\qquad\qquad\quad
    \includegraphics[scale=1]{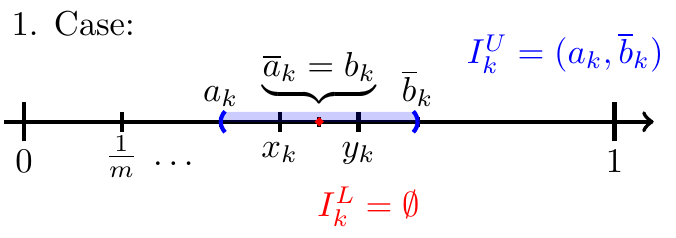}
    \vspace*{1ex}
    \qquad\qquad\qquad\quad 
    \includegraphics[scale=1]{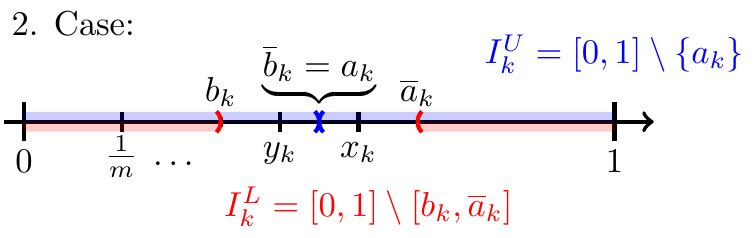}
    \vspace*{1ex}
    \qquad\qquad\qquad\quad 
    \includegraphics[scale=1]{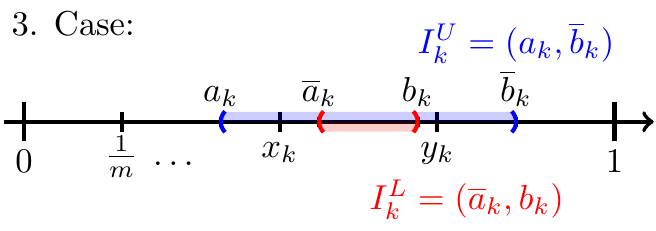}
        \vspace*{1ex}
    \qquad\qquad\qquad\quad 
    \includegraphics[scale=1]{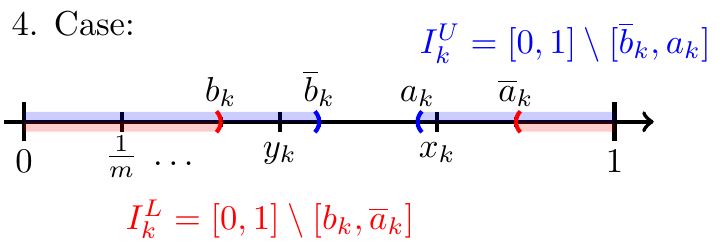}
    \caption{
  The four cases from the proof of Lemma~\ref{lem: delta_cov_B_per}
  to show the existence of $I_k^L, I_k^U$ such that
  $I_k^L\subseteq I_k(\bsx,\bsy) \subseteq I_k^U$
  and $\lambda_1(I_k^U\setminus I_k^L)\leq 2/m$ are illustrated.
  }\label{fig:cases}
\end{figure}
In all cases we have $I_k^L\subseteq I_k(\bsx,\bsy) \subseteq I_k^U$ 
 as well as $\lambda_1(I_k^U\setminus I_k^L)\leq 2/m$.
For
$
L_B = \Pi_{i=1}^d I_i^L \in \Gamma_\delta$ and $U_B = \Pi_{i=1}^d I_i^U \in \Gamma_\delta
$
the inclusion property with respect to $B(x,y)$ does hold and 
\begin{align*}
 \lambda_d(U_B\setminus L_B) & 
 = \Pi_{i=1}^d \lambda_1(I_i^U) - \Pi_{i=1}^d \lambda_1(I_i^L)\\
& = \sum_{k=1}^d \left[\Pi_{i=1}^{k-1} \lambda_1(I_i^L) 
(\lambda_1(I_k^U)-\lambda_1(I_k^L)) 
\Pi_{i=k+1}^d \lambda_1(I_i^U)\right] \leq \frac{2d}{m}.
\end{align*}
By the choice of $m$ the right-hand side $2d/m$ is bounded by $\delta$ and the assertion is proven.
\end{proof}

Now we easily can prove an upper bound of the minimal dispersion according 
to $\mathcal{B}_{\rm per}$ as formulated in Corollary~\ref{cor: cor_B_per}.

\begin{proofof}{Corollary~\ref{cor: cor_B_per}}
By the previous lemma we know that there is a $\delta$-cover
with cardinality bounded by $(4d\delta^{-1})^{2d}$.
Then by Corollary~\ref{cor1} with $c_1=4$, $c_2=1$ and $c_3=2$ the proof is finished.
\end{proofof}

\section{Conclusion}
Based on $\delta$-covers we provide
in the main theorem an
estimate of the
minimal dispersion similar to the one of \eqref{eq: vc_est}. 
In the case where 
the VC-dimension of the set of test sets is not known, but a suitable $\delta$-cover can be
constructed our Theorem~\ref{thm: mainthm} leads to new results, as illustrated for $\mathcal{B}_{\rm per}$. 
One might argue, that we only show existence of ``good'' point sets. However, 
Remark~\ref{rem} tells us that a uniformly distributed random point set has small dispersion with high probability.
As far as we know, an explicit
construction of such point sets is not known.

\begin{acknowledgement}
The author thanks Aicke Hinrichs, David Krieg, Erich Novak and Mario Ullrich
for fruitful discussions to this topic.
\end{acknowledgement}

%
%

\end{document}